\pdfoutput=1 %
\let\accentvec\vec
\documentclass[runningheads, envcountsame, a4paper]{llncs}

\let\vec\accentvec

\usepackage[T1]{fontenc} %
\usepackage{amssymb} %
\usepackage{mathtools} %
\usepackage[ruled,linesnumbered,resetcount,noend]{algorithm2e}
\usepackage[hidelinks]{hyperref}

\usepackage{tikz} %
\usetikzlibrary{automata,positioning,fit}

\newcommand{\ud}{\triangleq} %
\newcommand{\udr}{\stackrel{{\mbox{\tiny\ensuremath{\triangle}}}}{\Longleftrightarrow}} %
\newcommand{\gram}{{\mathcal{G}}} %
\newcommand{\auto}{{\mathcal{A}}} %
\newcommand{\Lang}{{\mathcal{L}}} %

\DeclarePairedDelimiter\set{\{}{\}}
\DeclareMathOperator{\lfp}{lfp}

\newcommand{\upw}{{\uparrow}} %

\makeatletter
\renewcommand\paragraph{\@startsection{paragraph}{4}{\z@}%
                     {-12\p@ \@plus -4\p@ \@minus -4\p@}%
                     {-0.5em \@plus -0.22em \@minus -0.1em}%
                     {\normalfont\normalsize\bfseries}}
\makeatother

\usepackage[textsize=scriptsize]{todonotes} %

\begin{document}

\title{A Uniform Framework for Language Inclusion Problems}
\author{
	Kyveli~Doveri\inst{1,2}\orcidID{0000-0001-9403-2860} \and
	Pierre~Ganty\inst{1}\orcidID{0000-0002-3625-6003} \and
	Chana Weil-Kennedy\inst{1}\orcidID{0000-0002-1351-8824}
}
\institute{
	IMDEA Software Institute, Madrid, Spain\\\email{firsname.lastname@imdea.org} \and Universidad Politécnica de Madrid, Spain
}
\authorrunning{K.\ Doveri, P.\ Ganty, and C.\ Weil-Kennedy}
\maketitle

\begin{abstract}
	We present a uniform approach for solving language inclusion problems.
	Our approach relies on a least fixpoint characterization and a quasiorder to compare words of the “smaller” language, reducing the inclusion check to a finite number of membership queries in the “larger” language.
	We present our approach in detail on the case of inclusion of a context-free language given by a grammar into a regular language.
	We then explore other inclusion problems and discuss how to apply our approach.
\end{abstract}

\section{Introduction}
We are interested in the classical problem of language inclusion, which given two language acceptors, asks whether the language of one acceptor is contained into the language of the other one.
This problem is traditionally solved by complementing the “larger” language, intersecting with the “smaller” language and checking for emptiness.
Here, we avoid explicit complementation.
We present a simple and uniform approach for deciding language inclusion problems based on the use of quasiorder relations on words.
In this paper we are interested in the decidable cases of this problem, where the underlying alphabet of the language is a finite set of symbols.
Even though we focus mainly on words of finite length our approach also applies to languages of infinite words.
At its core, our approach relies on two notions: a \emph{fixpoint characterization} of the “smaller” language and a \emph{quasiorder} to compare words.
Intuitively, the language inclusion algorithms we derive leverage the fixpoint characterization to compute increasingly many words of the “smaller” language.
After a finite amount of time the computation is stopped.
The algorithm then tests whether each of the computed words of the “smaller” language also belong to the “larger” one.
If one word fails the test, the inclusion does not hold and we have a counterexample.
Whether we can correctly conclude that inclusion holds depends on whether we have computed “enough” words.
The r\^ole of the quasiorder is precisely that of detecting when “enough” words have been computed.
It must satisfy two properties: the first one ensures that it takes only a finite amount of time to compute “enough” words, and the second one ensures that if there exists a counterexample to inclusion, then it has been computed.

We present our approach in detail on the case of inclusion of a context-free language given by a grammar into a regular language, an EXPTIME-complete problem~\cite{kasaiProblemsFormalLanguage1992}.
This case is simple yet non-trivial.
After presenting our approach in Section \ref{sec:framework}, we give distinct quasiorders that can be used in the decision procedure in Section \ref{sec:instantiation}.
Section~\ref{sec:aspects} then dives into the algorithmic aspects related to using the so-called state-based quasiorders.
The state-based quasiorders enable some modifications of the inclusion algorithm ultimately leading to the so-called antichains algorithms~\cite{dewulfAntichainsNewAlgorithm2006,holikAntichainsVerificationRecursive2015}.
We also show how the saturation approach put forward by Esparza et al.~\cite{esparzaUniformFrameworkProblems2000} can be leveraged for the particular inclusion problem of a straight-line program\footnote{Straight-line program are context-free grammars where at most one word is derived from each grammar variable.
} into a finite state automaton.
In Section \ref{sec:ouverture} we talk about the other language classes to which our framework can be applied.
We revisit the case where the two languages are given by finite state automata.
Next, we investigate the case asking whether the trace set of a finite process (which is a regular language) is contained into the trace set of a Petri net, a case that was first solved by Esparza et al.~\cite{jancarPetriNetsRegular1999}.
To demonstrate the generality of the approach, we survey how it can also be leveraged for the case of two languages of infinite words accepted by Büchi automata.
We finish by briefly mentioning a few more cases that can be tackled using our approach.

This paper works as an overview of our previous work on the topic, providing a simplified explanation and pointers to the previous papers.
In particular, the section presenting our approach greatly simplifies the framework put forward in \cite{GantyRV2021}.

\section{Preliminaries}
\label{sec:preliminaries}

\subsubsection{Well-Quasiorders, Complete Lattices and Kleene Iterates}

A \emph{quasiorder} (qo) on a set \(E\), is a binary relation \(\ltimes\; \subseteq E\times E\) that is reflexive and transitive.
A quasiorder \(\ltimes\) is a \emph{partial order} when \(\ltimes\) is antisymetric (\(x \ltimes y\wedge y \ltimes x \implies x=y\)).
A \emph{complete lattice} is a set \(E\) and a partial order \(\ltimes\) on \(E\) such that every subset \(X\subseteq E\) has a least upper bound (the \emph{supremum}) in \(E\).

A sequence \(\{s_{n}\}_{n\in\mathbb{N}}\in E^{\mathbb{N}}\) on quasiordered set \((E, \ltimes)\) is \emph{increasing} if for every \(n\in \mathbb{N}\) we have \(s_{n}\ltimes s_{n+1}\).
For a function \(f\colon E \to E\) on a quasiordered set \((E, \ltimes)\) and for all \(n\in \mathbb{N}\), we define the \(n\)-th iterate \(f^{n}\colon E \to E\) of \(f\) inductively as follows: \(f^{0} \ud \lambda x\ldotp x\); \(f^{n+1} \ud f \circ f^{n}\).
The denumerable sequence of \emph{Kleene iterates} of \(f\) starting from the bottom value \(\bot\in E\) is given by \(\{f^{n}(\bot)\}_{n\in \mathbb{N}}\).
A \emph{fixpoint} of \(f\) is \(x\) such that \(f(x)=x\).
The \emph{least fixpoint} of \(f\) is the smallest fixpoint of \(f\) with respect to \(\mathord{\ltimes}\), if \(f\) has at least one fixpoint; it is denoted \(\lfp f\).
Recall that when \((E, \ltimes)\) is a complete lattice and \(f\colon E \to E\) is an monotone function (i.e. \(d\ltimes d'\implies f(d)\ltimes f(d'))\) then it follows from the Knaster–Tarski theorem that \(f\) has a least fixpoint \(\lfp f\) which, provided \(f\) is continuous\footnote{\(f\) is continuous if{}f \(f\) preserves least upper bounds of nonempty increasing chains.}, is given by the supremum of the increasing sequence of Kleene iterates of \(f\).

Given a quasiorder \(\mathord{\ltimes}\) on a set \(E\), and \(X \subseteq E\) a subset, let \(\upw_\ltimes X\) denote the upward closure of \(X\) with respect to \(\mathord{\ltimes}\) given by \(\set{y \in E\mid \exists x\in X.
	x\ltimes y}\).
A quasiorder \(\ltimes\) on \(E\) is a \emph{well-quasiorder} (wqo) if for every infinite sequence \(\{S_n\}_{n\in\mathbb{N}}\in\wp(E)^{\mathbb{N}}\) such that \(\upw_\ltimes S_1 \subseteq \upw_\ltimes S_2 \subseteq \cdots \) we have that there exists \(i\in\mathbb{N}\) such that \(\upw_\ltimes S_{i} \supseteq \upw_\ltimes S_{i+1}\), namely, \(\mathord{\ltimes}\) is a wqo if{}f there is no infinite strictly increasing chain of upward closed subsets in \(E\).

There are other equivalent definitions for wqo but the one using upward closed sets is the most convenient for our purpose.
One property of wqo we will leverage throughout the paper is that the component wise lifting of a wqo remains a wqo, namely, given \(n\in\mathbb{N}\), if \(\mathord{\ltimes}\) is a wqo then so is \(\mathord{\ltimes^n}\).

Finally, we introduce the lifting of the qo \(\mathord{\ltimes}\) to sets by defining \(\mathord{\sqsubseteq_{\ltimes}} \subseteq \wp(E) \times \wp(E)\) as follows: \[X \sqsubseteq_{\ltimes} Y \udr \forall x\in X, \exists y\in Y, y\ltimes x\enspace.
\]
It is routine to check that \(\mathord{\sqsubseteq_{\ltimes}}\) is a quasiorder as well and also that \(X \sqsubseteq_{\ltimes}
Y\) if{}f \( \upw_\ltimes X \subseteq \upw_\ltimes Y\).
Given the previous equivalence, the rationale to introduce \(\mathord{\sqsubseteq_{\ltimes}}\) is algorithmic since \(\mathord{\sqsubseteq_{\ltimes}}\) is straightforward to implement for finite sets \(X,Y\) given a decision procedure for \(\mathord{\ltimes}\); whereas \(\upw_\ltimes X \subseteq \upw_\ltimes Y\) does not give a straightforward implementation even when \(X\) and \(Y\) are finite and \(\mathord{\ltimes}\) is decidable.

\subsubsection{Alphabets, Words and Languages}

An \emph{alphabet} is a nonempty finite set of symbols, generally denoted by \(\Sigma\).
A \emph{word} is a sequence of symbols over the alphabet \(\Sigma\).
The set of finite words and the set of infinite words over \(\Sigma\) are denoted by \(\Sigma^{*}\) and \(\Sigma^{\omega}\) respectively.
We denote by \(\varepsilon\) the \(\emph{empty word}\) and define \(\Sigma^{+}\ud \Sigma^{*}\backslash \{\varepsilon\}\).
A \emph{language of finite words} over \(\Sigma\) is a subset of \( \Sigma^{*}\).
A \emph{language of infinite words} or \(\omega\)-\emph{language} over \(\Sigma\) is a subset of \( \Sigma^{\omega}\).

\subsubsection{Finite Automata}

A \emph{finite automaton} (FA) on an alphabet \(\Sigma\) is a tuple \(\auto=(Q,\delta,q_{I},F)\) where \(Q\) is a finite set of states including an initial state \(q_{I}\in Q\), \(\delta\colon Q\times \Sigma \rightarrow \wp(Q)\) is a transition function, and \(F\subseteq Q\) is a subset of final states.
Let \(q \overset{a}\rightarrow q'\) denote a transition \(q' \in \delta(q,a)\) that we lift to finite words by transitive and reflexive closure, thus writing \(q \xrightarrow{u}\!
\!^* q'\) with \(u\in\Sigma^*\).
The language of finite words accepted by \(\auto\) is \(L^{*}(\auto) \ud \textstyle \set{u\in \Sigma^{*} \mid \exists q\in F.
	q_{I} \xrightarrow{u}\!\!^* q}\).
An \emph{accepting trace} of \(\auto\) on an infinite word \(w = a_0a_1\cdots \in \Sigma^\omega\) is an infinite sequence \(q_0 \overset{a_0}\rightarrow q_1 \overset{a_1}\rightarrow q_2 \cdots\) such that \(q_{0}=q_{I}\) and \(q_{j} \in F\) for infinitely many \(j\)’s.
The \(\omega\)-language accepted by \(\auto\) is \(L^{\omega}(\auto)\ud \{ \xi \in \Sigma^\omega \mid \text{ there is an accepting trace of } \auto \text{ on } \xi\}\).
We call \(\auto\) a Büchi automaton (BA) when we consider it as an acceptor of infinite words.
A language \(L \subseteq \Sigma^{\omega}\) is \(\omega\)-\emph{regular} if \(L=L^{\omega}(\auto)\) for some BA~\(\auto\).

\subsubsection{Context-free Grammars}

A \emph{context-free grammar} (CFG) or simply \emph{grammar} on \(\Sigma\) is a tuple \(\gram=(V,P)\) where \(V=\{X_{1},\ldots,X_{n}\}\) is the finite set of variables, and \(P\) is the finite set of production rules \(X_{i} \rightarrow \beta\) where \(\beta \in (V\cup \Sigma)^{*}\).
Given \(\alpha,\alpha' \in (V\cup \Sigma)^{*}\), we write \(\alpha \Rightarrow \alpha'\) if there exists \(\gamma, \delta \in (V\cup \Sigma)^*\) and \(X_i \rightarrow \beta \in P\) such that \(\alpha = \gamma X_i \delta\) and \(\alpha' = \gamma \beta \delta\).
The reflexive and transitive closure of \(\mathord{\Rightarrow}\) is written \(\mathord{\Rightarrow^*}\).
For \(i\in [1,n]\), let \(L_i(\mathcal{G}) = \set{w \in \Sigma^* \mid X_i \Rightarrow^* w}\).
When we omit the subscript \(i\) it is meant to be \(1\), hence we define the \emph{language accepted} by \(\gram\) as \(L(\gram)\), that is \(L_1(\gram)\) given by \(\set{w \in \Sigma^* \mid X_1 \Rightarrow^* w}\).
A CFG is in Chomsky normal form (CNF) if all of its production rules are of the form \(X_{j} \rightarrow X_{k}\,X_{l}\), \(X_{j} \rightarrow a\), or \(X_{1}\rightarrow \varepsilon\) where \(X_{j},X_{k},X_{l}\in V\) and \(a\in \Sigma\).
\begin{example}
	\label{ex:cfg}
	Let \(\gram'=(V',P')\) be the CFG on \(\Sigma=\set{a,b}\) given by \(V'=\set{X_1}\) and \(P'=\set{X_1 \rightarrow \varepsilon, X_1 \rightarrow a\, X_1\, b}\).
	Notice that this grammar is not in CNF.
	The language \(L(\gram')\) is \(\set{a^n b^n \mid n \ge 0}\).
	We can define a grammar \(\gram=(V,P)\) in CNF such that \(L(\gram)=L(\gram')\).
	Let \(V = \set{X_1, X_2, X_3, X_4}\) and \(P\) be the rules
	\begin{align*}
		X_1 & \rightarrow \varepsilon & X_2 & \rightarrow a & X_1 & \rightarrow X_2\, X_3 \\ && X_4 &\rightarrow b & X_3 &\rightarrow X_1\,X_4 \enspace.
	\end{align*}
\end{example}

\section{Algorithm}
\label{sec:framework}

In this section we outline the quasiorder-based approach to solving the inclusion problem of a context-free language into a regular language, which reduces the problem to finitely many membership queries.
A membership query is a check of whether a given word belongs to a given language.
More precisely, we consider the inclusion problem \(L(\gram) \subseteq M\) where \(\gram\) is a CFG and \(M\) is a regular language.
To solve the inclusion problem, we select a subset \(S\) of words of \(L(\gram)\) such that (i) \(S\) is finite, (ii) \(S\) is effectively computable, and (iii)
\(S\) contains a counterexample to \(L(\gram)\subseteq M\) if the inclusion does not hold.
Upon computing such a set \(S\) the inclusion check \(L(\gram)\) into \(M\) reduces to finitely many membership queries of the words of \(S\) into \(M\).

More concretely, the computation of \(S\) will be guided by a quasiorder \(\mathord{\ltimes}\) on words.
The set \(S\) of words is such that, as per the quasiorder \(\mathord{\ltimes}\), the following holds: \(S\subseteq L(\gram)\) and \(L(\gram) \sqsubseteq_{\ltimes} S\).
\footnote{We can even relax the inclusion \(S\subseteq L(\gram)\) to the weaker condition \(S \sqsubseteq_{\ltimes}
	L(\gram)\).
}
In that setting, for \(S\) to satisfy (i), we will require (1) \(\mathord{\ltimes}\) to be a well-quasi order.
Moreover for \(S\) to comply with (ii), we will require (2) \(\mathord{\ltimes}\) to be decidable but also ``monotonic'' (it will become clear what it means and why we need this condition later).
Finally, for \(S\) to comply with (iii), we first formalize the requirement:
\begin{equation}
	\tag{\(\star\)}
	\label{eq:reduction}
	L(\gram) \subseteq M \iff \forall u\in S, ~u\in M \enspace.
\end{equation}
Intuitively, it means that whatever set \(S\) we select needs to be included in \(M\) if the inclusion \(L(\gram) \subseteq M\) holds and otherwise the set \(S\) needs to contain at least one counterexample to the inclusion.
To achieve this, we will require (3) \(\ltimes\) to be \(M\)-preserving: A quasiorder \(\mathord{\ltimes} \subseteq \mathord{\Sigma^{*}\times \Sigma^{*}}\) is said to be \(M\)-\emph{preserving} when for every \(u,v \in \Sigma^{*}\) if \(u \in M\) and \(u \ltimes v\) then \(v \in M\)\footnote{This is equivalent to saying that M is upward-closed w.r.t. the quasiorder \(\mathord{\ltimes} \).}.
To see how \(\ltimes\) being \(M\)-preserving enforces \eqref{eq:reduction} let us first assume that \(L(\gram)\subseteq M\) holds.
Since \(S\subseteq L(\gram)\) and \(L(\gram)\subseteq M\), the direction \(\mathord{\Rightarrow}\) of Equation~\eqref{eq:reduction} holds.
This still holds if we have the weaker condition \(S \sqsubseteq_{\ltimes} L(\gram)\).
In this case, for every \(s\in S\) there exists \(u\in L(\gram)\) such that \(u\ltimes s\), hence the assumption \(L(\gram)\subseteq M\) together with the \(M\)-preservation show that \(s\in M\).
Now assume that the right-hand side of \eqref{eq:reduction} holds and let \(u\in L(\gram)\).
Since \(L(\gram)\sqsubseteq_{\ltimes} S\) there is \(s\in S\) such that \(s\ltimes u\).
Since \(s\in M\), \(s\ltimes u\) and \(\ltimes\) is \(M\)-preserving we have \(u\in M\).

To compute \(S\) we leverage a fixpoint characterization of \(L(\gram)\).
Then we effectively compute the set \(S\) by computing finitely many Kleene iterates of the fixpoint characterization of \(L(\gram)\) so as to obtain a set \(S\) such that \(L(\gram) \sqsubseteq_{\ltimes} S\).
In order to decide we have computed enough Kleene iterates we rely on \(\sqsubseteq_{\ltimes}\) (which is why we need \(\ltimes\) to be decidable).
We also need \(\ltimes\) to satisfy a so-called monotonicity condition for otherwise we cannot guarantee that \(S\) satisfies \(L(\gram) \sqsubseteq_{\ltimes} S\).

Let us now turn to the characterization of \(L(\gram)\) as the least fixpoint of a function.
Fix a grammar \(\gram\) in CNF and \(M\) a regular language.
Our algorithms also work when the grammar \(\gram\) is not given in CNF, we assume CNF for the simplicity of the presentation.

\paragraph{Least Fixpoint Characterization.}
To compute our finite set \(S\) we use a characterization of \(L(\gram)\) as the least fixpoint of a function, and show that we can iteratively compute the function's Kleene iterates until we compute a set \(S\) such that \(L(\gram)\sqsubseteq_{\ltimes} S\).
\begin{example}
	\label{ex:iterations}
	Take the grammar \(\gram\) of Example~\ref{ex:cfg} such that \(L(\gram)=\set{a^n b^n \mid n \ge 0}\).
	We define the function \(F_{\gram}\colon \wp(\Sigma^*)^4 \rightarrow \wp(\Sigma^*)^4\) where \((L_1,L_2,L_3,L_4)\in \wp(\Sigma^*)^{4}\) is a vector of languages of finite words as follows: \[ F_{\gram}\colon (L_1,L_2,L_3,L_4) \mapsto (L_2L_3 \cup \set{\varepsilon}, \set{a}, L_1L_4, \set{b})\enspace .
	\]
	If we apply  \(F_{\gram}\) to the empty vector \((\emptyset, \emptyset, \emptyset, \emptyset)\in \wp(\Sigma^*)^{4}\),
	we get the vector of languages \((\set{\varepsilon}, \set{a}, \emptyset, \set{b})\)
	(recall that \(L\emptyset = \emptyset\) and \(\emptyset L = \emptyset\) for any language \(L\)).
	Applying \(F_{\gram}\) to this last vector we obtain \((\set{\varepsilon}, \set{a}, \set{b}, \set{b})\), i.e. \(F_{\gram}^2(\emptyset, \emptyset, \emptyset, \emptyset)=(\set{\varepsilon}, \set{a}, \set{b}, \set{b})\).
	We give a list of the first few repeated applications of \(F_{\gram}\) to the empty vector.
	\begin{align*}
		F_{\gram}(\vec{\emptyset}) & = (\set{\varepsilon}, \set{a}, \emptyset, \set{b}) \\ F_{\gram}^2(\vec{\emptyset}) &= (\set{\varepsilon}, \set{a}, \set{b}, \set{b}) \\ F_{\gram}^3(\vec{\emptyset}) &= (\set{ab,\varepsilon}, \set{a}, \set{b}, \set{b}) \\ F_{\gram}^4(\vec{\emptyset}) &= (\set{ab,\varepsilon}, \set{a}, \set{ab^2,b}, \set{b}) \\ F_{\gram}^5(\vec{\emptyset}) &= (\set{a^2b^2,ab,\varepsilon}, \set{a}, \set{ab^2,b}, \set{b}) \\ F_{\gram}^6(\vec{\emptyset}) &= (\set{a^2b^2,ab,\varepsilon}, \set{a}, \set{a^2b^3,ab^2,b}, \set{b}) \\ F_{\gram}^7(\vec{\emptyset}) &= (\set{a^3b^3, a^2b^2,ab,\varepsilon}, \set{a}, \set{a^2b^3,ab^2,b}, \set{b})\enspace .
	\end{align*}
	It is not hard to see that for all words \(a^nb^n \in L(\gram)\), there exists an \(i\) such that \(a^nb^n\) appears in the first component of \(F_{\gram}^i(\vec{\emptyset})\).
	In fact \(i=2n+1\).
\end{example}
We now formally define the function \(F_{\gram}\) over \(\wp(\Sigma^*)^{n}\), i.e. over \(n\)-vectors of languages of finite words over \(\Sigma\), where \(n\) is the number of variables of our grammar \(\gram\).
Let \(\Lang=(L_1,\ldots,L_n)\in \wp(\Sigma^*)^{n}\) be an \(n\)-vector of languages \(L_1, \dots, L_n\).
For each \(j\in [1,n]\), the \(j\)-th component of \(F_{\gram}(\Lang)\), denoted \(F_{\gram}(\Lang)_j\), is defined as: \[F_{\gram}(\Lang)_j \ud \bigcup_{X_{j}\rightarrow X_{k}\, X_{k'}\in P} L_{k}\, L_{k'} \cup \bigcup_{a \in \Sigma \cup \{\varepsilon\},~X_{j}\rightarrow a\in P}\set{a}\enspace .
\]
We have the following least fixpoint characterization for \(\gram\).
\begin{proposition}[Least fixpoint characterization]
	\label{prop:lfp}
	For all \(i\in\set{1,\ldots,n}\) we have \(L_i(\gram)=(\lfp F_{\gram})_{i}\).
\end{proposition}
\begin{proof}
	The function \(F_{\gram}\) is increasing and the supremum of the increasing sequence of its Kleene iterates starting at the bottom value \((\emptyset, \ldots, \emptyset)\in \wp(\Sigma^*)^{n}\) is the vector equal to \(\set{u\in \Sigma^{*} \mid X_{j} \Rightarrow^{*} u}\) for each \(j \in [1,n]\).
	Therefore, by the Knaster–Tarski theorem applied to the complete lattice \((\wp(\Sigma^*)^{n},\subseteq^{n})\) and \(F_{\gram}\), \(\lfp F_{\gram}=(\set{u\in \Sigma^{*} \mid X_{j} \Rightarrow^{*} u})_{j \in [1,n]}\).
	Thus, \((\lfp F_{\gram})_{i} =L_i(\gram)\).
	\qed
\end{proof}
\begin{remark}
	The function definition uses the fact that \(\gram\) is given in CNF, notably that its rules are of the form \(X_{j}\rightarrow X_{k}X_{k'}\).
	However a function \(F'_{\gram}\) can be defined for a grammar \(\gram\) in any form such that \((\lfp F'_{\gram}) = (L_1(\gram),\ldots,L_n(\gram))\) and if \(\gram\) is in CNF then \(F'_{\gram}\) and \(F_{\gram}\) coincide, see \cite[Definition 2.9.1 and Theorem 2.9.3]{GallierTOC}.
\end{remark}

The ordering \(\mathord{\sqsubseteq_{\ltimes^{n}}}\) is used to compare the Kleene iterates of the function \({F_{\gram}}\).
For it to be apt to detect convergence, hence when the algorithm should stop, the quasiorder \(\ltimes\) needs to be monotonic.
A quasiorder is \emph{monotonic} if it is right-monotonic and left-monotonic.
A quasiorder \(\ltimes\) on \(\Sigma^*\) is \emph{right-monotonic} (respectively \emph{left-monotonic}) if for all \(u,v\in \Sigma^*\), for all \(a\in \Sigma\), \(u\ltimes v\) implies \(ua \ltimes va\) (respectively \(u\ltimes v\) implies \(au \ltimes av\)).

\begin{proposition}[\(\sqsubseteq_{\ltimes}\) monotonicity]
	\label{prop:antichain_opt}
	Let \(\ltimes \) be a monotonic quasiorder on \(\Sigma^{*}\).
	If \(Y \sqsubseteq_{\ltimes^{n}} S\) then \(F_{\gram}(Y) \sqsubseteq_{\ltimes^{n}} F_{\gram}(S) \).
\end{proposition}
\begin{proof}
	Fix \(j \in [1,n]\).
	We want to show that for all \(y \in F_{\gram}(Y)_j\) there exists \(s \in F_{\gram}(S)_j\) such that \(s \ltimes y\).
	Let \(y \in F_{\gram}(Y)_j\).
	If there is an \(a\in \Sigma \cup \{\varepsilon\}\) such that \(X_{j}\rightarrow a\in P\) and \(y=a\), then \(y\) is also in \(F_{\gram}(S)_j\) by definition.
	Since \(\ltimes\) is a quasiorder it is reflexive, so \(y \ltimes y\).
	Otherwise, by definition of \(F_{\gram}(Y)_j\), there exist \(y_k\in Y_k, y_{k'} \in Y_{k'}\) such that \(X_{j}\rightarrow X_{k}X_{k'}\in P\) and \(y=y_k y_{k'}\).
	Since \(Y\sqsubseteq_{\ltimes^n} S\), there exist \(s_k\in S_k, s_{k'} \in S_{k'}\) such that \(s_k\ltimes y_{k}\) and \(s_{k'} \ltimes y_{k'}\).
	Using \(X_{j}\rightarrow X_{k}X_{k'}\in P\) and the definition of \(F_{\gram}(S)_j\), we have \(s_k s_{k'} \in F_{\gram}(S)_j\).
	By right-monotonicity of our quasiorder, \(s_k s_{k'}\ltimes y_k s_{k'}\).
	By left-monotonicity, \(y_k s_{k'}\ltimes y_k y_{k'}\).
	Since quasiorders are transitive, we obtain \(s_k s_{k'}\ltimes y_k y_{k'}\).
	We thus conclude from the above that \(F_{\gram}(S)_j \sqsubseteq_{\ltimes} F_{\gram}(Y)_j\), hence it is easy to see that \(F_{\gram}(S) \sqsubseteq_{\ltimes^n} F_{\gram}(Y)\).
	\qed
\end{proof}

Under the assumption of monotonicity on our \(M\)-preserving well-quasiorder we can iteratively compute a finite set \(S\) such that \(L(\gram)\sqsubseteq_{\ltimes} S\) using the function \({F_{\gram}}\) as shown by the next proposition.
\begin{proposition}[Stabilization]
	\label{prop:cv_check}
	Let \(\ltimes\) be a wqo on \(\Sigma^{*}\).
	There is an \(m\ge 0\) such that \({F_{\gram}}^{m+1}(\vec{\emptyset}) \sqsubseteq_{\ltimes^n} {F_{\gram}}^{m}(\vec{\emptyset})\); and if \(\ltimes\) is monotonic then \(\lfp F_{\gram} \sqsubseteq_{\ltimes^{n}} {F_{\gram}}^{m}(\vec{\emptyset}) \).
\end{proposition}
\begin{proof}
	For simplicity, we write this proof for the case where \(n\), the number of variables in \(\gram\), is equal to \(1\).
	It holds also for \(n> 1\) by reasoning component-wise.
	Since \(\ltimes\) is a wqo on \(\Sigma^{*}\), the set of upward closed sets of \(\wp(\Sigma^{*})\) ordered by inclusion has the ascending chain condition \cite[Theorem 1.1]{LucaV94}, meaning there is no infinite strictly increasing sequence of upward closed sets.
	It is routine to check by induction using Proposition~\ref{prop:antichain_opt} that \({F_{\gram}}^{i}(\emptyset) \sqsubseteq_{\ltimes} {F_{\gram}}^{i+1}(\emptyset)\) for all \(i\ge 0\).
	Therefore also \(\upw_\ltimes {F_{\gram}}^{i}(\emptyset) \subseteq \upw_\ltimes{F_{\gram}}^{i+1}(\emptyset)\) for all \(i\ge 0\).
	By the ascending chain condition, there exists a positive integer \(m\) such that \(\upw_\ltimes {F_{\gram}}^{m+1}(\emptyset) \subseteq \upw_\ltimes {F_{\gram}}^{m}(\emptyset)\), which is equivalent to \({F_{\gram}}^{m+1}(\emptyset) \sqsubseteq_{\ltimes} {F_{\gram}}^{m}(\emptyset)\).

	Assume that \(\ltimes\) is monotonic.
	An induction using Proposition~\ref{prop:antichain_opt} and the above shows that for every \(k \geq m\), \({F_{\gram}}^{k+1}(\emptyset) \sqsubseteq_{\ltimes} {F_{\gram}}^{k}(\emptyset)\).
	Hence, by transitivity of \(\sqsubseteq_{\ltimes}\) we deduce that for every \(k \geq m\), \({F_{\gram}}^{k}(\emptyset)\sqsubseteq_{\ltimes}{F_{\gram}}^{m}(\emptyset)\).
	Recall that \(\lfp F_{\gram}\) is the supremum of the sequence of Kleene iterates of \(F_{\gram}\), i.e. of \(\{ F_{\gram}^{n} (\emptyset) \}_{n\in \mathbb{N}}\).
	In the complete lattice \((\wp(\Sigma^*)^{n},\subseteq^{n})\), the supremum of \(\{ F_{\gram}^{n} (\emptyset) \}_{n\in \mathbb{N}}\) is equal to \(\bigcup_{n\in \mathbb{N}} F_{\gram}^{n} (\emptyset)\).
	Thus \(\lfp F_{\gram}\sqsubseteq_{\ltimes} {F_{\gram}}^{m}(\emptyset)\).
	\qed
\end{proof}

\paragraph{Algorithm.}
Given a regular language \(M\subseteq \Sigma^{*}\) we say that a quasiorder \(\mathord{\ltimes}\subseteq \mathord{ \Sigma^{*}\times \Sigma^{*}}\) is \emph{\(M\)-suitable} if it is 1) a wqo, 2) \(M\)-preserving, 3) monotonic and 4) decidable i.e., given two words \(u\) and \(v\) we can decide whether \(u\ltimes v\).
Intuitively a quasiorder is \(M\)-suitable if it can be used in our quasiorder-based framework to decide the inclusion problem \(L(\gram) \subseteq M\).
\begin{algorithm}[hbt]
	\label{algo:1}
	\KwData{\(\gram=(V,P)\) CFG with \(n\) variables.}%
	\KwData{\(M\)-suitable quasiorder \(\ltimes\).}
	\KwData{Procedure deciding  \(u \in M\) given \(u \).}
	\smallskip
	\(\text{cur} \coloneq \vec{\emptyset}\)\;
	\Repeat{\(\mathrm{cur} \sqsubseteq_{\ltimes^{n}} \mathrm{prev}\)}{
		\(\text{prev} \coloneq \text{cur}\)\;
		\(\text{cur} \coloneq {F_{\gram}}(\text{prev})\)\;
	}
	\ForEach{\(u\in (\mathrm{prev})_{1}\)}{
		\lIf{\( u \notin M \)}{\Return{\(\mathsf{false}\)}}
	}
	\Return{\(\mathsf{true}\)};
	\caption{Algorithm for deciding \(L(\gram)\subseteq M\)
		\label{algo}}
\end{algorithm}
\begin{theorem}
	Algorithm~\ref{algo} decides the inclusion problem \(L(\gram)\subseteq M\).
\end{theorem}
\begin{proof}
	As established by Proposition~\ref{prop:cv_check}, given a monotonic decidable wqo \(\ltimes\), lines \(2\) to \(5\) of Algorithm~\ref{algo} compute, in finite time, a finite set equal to \({F_{\gram}}^{m}(\vec{\emptyset})\) for some \(m\) such that \(\lfp F_{\gram}\sqsubseteq_{\ltimes^n}{F_{\gram}}^{m}(\vec{\emptyset})\).
	Moreover, since \({F_{\gram}}^{m}(\vec{\emptyset})\subseteq \lfp F_{\gram}\), we find that \({F_{\gram}}^{m}(\vec{\emptyset})\sqsubseteq_{\ltimes}\lfp F_{\gram}\).
	Since \(\ltimes\) is \(M\)-preserving, Equation~\eqref{eq:reduction} holds for \({F_{\gram}}^{m}(\vec{\emptyset})_1\), i.e.
	\begin{align*}
		\label{eq:algo}
		L(\gram) \subseteq M & \iff \forall u\in {F_{\gram}}^{m}(\vec{\emptyset})_1,~u\in M \enspace .
	\end{align*}
	Thus the algorithm is correct, and it terminates because there are only finitely many membership checks \(u\in M\) in the for-loop at lines \(6\) and \(7\).
	\qed
\end{proof}
\begin{remark}
	The loop at lines 2 to 5 of the algorithm iteratively computes a vector of sets of words \(\text{cur}\), updating its value to \(F_{\gram}(\text{prev})\) at line 4.
	The algorithm remains correct if we remove some words from \(\text{prev}\) before each new update, as long as the resulting vector \(\text{prev}\) is included in \(\text{cur}\) with respect to \(\sqsubseteq_{\ltimes}\) (intuitively this ensures we do not lose any counter-examples).
	That is, we can replace the assignment line 3 by \(\text{prev} \coloneq \text{cur'}\) for any \(\text{cur'} \subseteq \text{cur}\) such that \(\text{cur} \sqsubseteq_{\ltimes} \text{cur'}\); the correctness follows from Proposition~\ref{prop:antichain_opt}.
\end{remark}

\section{Quasiorder Instantiation}
\label{sec:instantiation}

We instantiate the algorithm with a quasiorder and give an example run.
We then discuss other quasiorders with which the algorithm can be instantiated.

\subsection{A State-based Quasiorder}
\label{sec:example}

We present an \(M\)-suitable quasiorder to instantiate Algorithm~\ref{algo}.
The quasiorder is derived from a finite automaton with language \(M\).
Given an automaton \(\auto=(Q,q_{I},\delta,F)\) with \(L(\auto)=M\) and a word \(u\in \Sigma^*\), we define the \emph{context} set of \(u\) \[\text{ctx}^{\auto}(u)\ud\{(q,q')\in Q^2\mid q\xrightarrow {u}\!
	\!^*q'  \}\enspace .\]
We derive the following quasiorder on \(\Sigma^*\): \[u\leq_{\text{ctx}}^{\auto} v \udr \text{ctx}^{\auto}(u)\subseteq \text{ctx}^{\auto}(v)\enspace .
\]
\begin{proposition}
	The quasiorder \(\leq_{\text{ctx}}^{\auto}\) is \(M\)-suitable.
\end{proposition}
\begin{proof}
	The result \cite[Lemma 7.8]{GantyRV2021} shows that \(\leq_{\text{ctx}}^{\auto}\) is a \(M\)-preserving, monotonic wqo.
	Since for every \(u \in \Sigma^{*}\) we can compute the finite set \(\text{ctx}^{\auto}(u)\), we deduce that given two words \(u\) and \(v\) we can decide \(u\leq_{\text{ctx}}^{\auto}v\).
	Thus, \(\leq_{\text{ctx}}^{\auto}\) is \(M\)-suitable.
	\qed
\end{proof}
\begin{figure}[ht]
	\centering
	\begin{tikzpicture}[node distance=2cm, auto, initial text={}, transform shape]%
		\node[state, accepting, initial] (q0) {\(p\)};
		\node[state, accepting, right of=q0] (q1) {\(q\)};

		\path[->]
		(q0) edge[above=20] node {\(b\)} (q1)
		(q0) edge[loop above] node {\(a\)} ()
		(q1) edge[loop above] node {\(b\)} ()
		;
	\end{tikzpicture}
	\caption{Finite automaton \(\auto\) recognizing language \(a^*b^*\).}
	\label{fig:automaton-for-M}
\end{figure}
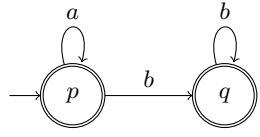
Let \(M\) be the regular language \(a^*b^*\).
Consider the automaton \(\auto\) depicted in \figurename~\ref{fig:automaton-for-M}.
It has two states \(p,q\) and accepts \(M\), i.e. \(L(\auto)=M\).
We give an execution of our algorithm instantiated with \(\leq_{\text{ctx}}^{\auto}\) to decide \(L\subseteq M\), where \(L\) is the language recognized by the grammar \(\gram\) of Example~\ref{ex:cfg}, that is \(L=L(\gram)=\set{a^n b^n \mid n \ge 0}\).
We compute the contexts of the words appearing in the first iterations of \(F_\gram\), as computed in Example~\ref{ex:iterations}.
\begin{align*}
	\text{ctx}^{\auto}(\varepsilon) & = \set{(p,p),(q,q)} & \text{ctx}^{\auto}(b) & = \set{(p,q),(q,q)}                                    \\
	\text{ctx}^{\auto}(a)           & = \set{(p,p)}       & \text{ctx}^{\auto}(w) & = \set{(p,q)},      &  & \forall w\in a^+b^+\enspace .
\end{align*}
The algorithm computes the first four iterations of \(F_{\gram}\) and then stops.
We have \(F_{\gram}^5(\emptyset, \emptyset, \emptyset, \emptyset)_1= \set{a^2b^2,ab,\varepsilon}\), \(F_{\gram}^4(\emptyset, \emptyset, \emptyset, \emptyset)_1= \set{ab,\varepsilon}\), and the languages of the other component are the same.
Since \(ab \leq_{\text{ctx}}^{\auto} a^2b^2\), we have that \(F_{\gram}^5(\emptyset, \emptyset, \emptyset, \emptyset)\sqsubseteq_{(\leq_{\text{ctx}}^{\auto})^4} F_{\gram}^4(\emptyset, \emptyset, \emptyset, \emptyset)\), hence \(\lfp F_{\gram}\sqsubseteq_{(\leq_{\text{ctx}}^{\auto})^4} F_{\gram}^4(\emptyset, \emptyset, \emptyset, \emptyset)\).
For each \(u \in F_{\gram}^4(\emptyset, \emptyset, \emptyset, \emptyset)_1\) we check if \(u \in a^*b^*\).
This is the case, and the algorithm returns \(\mathsf{true}\).
\subsection{Other Quasiorders}
\label{sec:other-qo}

The first kind of quasiorder we presented was a \emph{state-based} quasiorder derived from a finite automaton for the language \(M\).
Here we present a \emph{syntactic} quasiorder based on the syntactic structure of the language \(M\).
Given a word \(u\in \Sigma^*\), we define the set \[\text{ctx}^{M}(u)\ud\{(w,w')\in \Sigma^*\times \Sigma^* \mid wuw' \in M \}\enspace.
\]
We derive the following quasiorder on \(\Sigma^*\):\\ \[u\leq_{\text{ctx}}^{M} v \udr \text{ctx}^{M}(u)\subseteq \text{ctx}^{M}(v)\enspace.
\]
Following \cite{LucaV94}, we call this the Myhill order.
\begin{proposition}
	The quasiorder \(\leq_{\text{ctx}}^{M}\) is \(M\)-suitable.
	Moreover, it is the coarsest among the \(M\)-suitable quasiorders.
\end{proposition}
\begin{proof}
	Since \(M\) is a regular language, by Lemma 7.6 (a) in \cite{GantyRV2021} \(\leq_{\text{ctx}}^{\auto}\) is \(M\)-suitable.
	By Lemma 7.6 (b) if \(\ltimes \) is a \(M\)-suitable quasiorder then for every \(u\in \Sigma^{*}\) we have \(\upw_\ltimes \{u \} \subseteq \upw_{\leq_{\text{ctx}}^{M}} \{u\}\), in other words \(u\ltimes v\) implies \(u\leq_{\text{ctx}}^{M} v\).
	Hence, \(\leq_{\text{ctx}}^{M}\) is coarser than any \(M\)-suitable quasiorder.
	\qed
\end{proof}

Notice that given two words \(u\) and \(v\) deciding \(u\leq_{\text{ctx}}^{M} v\) reduces to an inclusion between two regular languages, it is therefore PSPACE-complete to decide it.
On the other hand deciding \(u\leq_{\text{ctx}}^{\auto}v\) is PTIME.

\section{Algorithmic Aspects}
\label{sec:aspects}

In Section~\ref{sec:example} we considered the so-called state-based quasiorders for Algorithm~\ref{algo:1}.
We can leverage state-based quasiorders even further and modify the algorithm to drop words entirely and replace them by their context.
As we will see, contexts (that is sets of pairs of states) carry sufficient information to perform all the tasks required by the algorithm (namely comparison between words, updates via word concatenation and word membership tests).

\subsection{A State-Based Variant}
As a first step, consider a variant of Algorithm~\ref{algo:1} where words are stored together with their context: instead of \(w\) we will store \((w,\text{ctx}^{\auto}(w))\).
This change preserves the logic of the algorithm, hence its correctness.
Observe that because the contexts are readily available they can be used directly for comparisons: instead of computing contexts each time two words need to be compared, the algorithm compares the already-computed contexts.
Furthermore, and this is key, contexts can be updated during computations because they can be characterized inductively.
Indeed, given \(\text{ctx}^{\auto}(w)\) and \(\text{ctx}^{\auto}(w')\), \(\text{ctx}^{\auto}(w w')\) can be computed as follows: \[ \text{ctx}^{\auto}(w w') = \{ (p,q) \mid \exists p'\colon (p,p') \in \text{ctx}^{\auto}(w) \land (p',q)\in \text{ctx}^{\auto}(w') \}\enspace.
\]
In addition, membership queries in \(M=L(\auto)\) have an equivalent counterpart on contexts:
\(w \in M\) if{}f \(\text{ctx}^\auto(w) \cap (\{q_I\}\times F) \neq\emptyset\).
As a consequence of the above, Algorithm~\ref{algo:1} can drop words entirely and focus exclusively on contexts.
We call this algorithm the antichains algorithm following a prolific history of such algorithms like \cite{holikAntichainsVerificationRecursive2015} (for the CFG into REG case) and starting with \cite{dewulfAntichainsNewAlgorithm2006} (for the REG into REG case).

Next, we will see a data-structure for contexts tailored to the case where \(L(\gram)\) is a singleton.
This data-structure, which leverages the automaton structure of regular languages, has been studied by Esparza, Rossmanith and Schwoon in a paper from 2000 \cite{esparzaUniformFrameworkProblems2000}.
Their paper list several potential use of the data-structure but not the one we are giving next.

\subsection{A Data-Structure for the Case of Straight Line Programs}
\label{sec:data-structure}

In 2000 Esparza, Rossmanith and Schwoon published a paper \cite{esparzaUniformFrameworkProblems2000} in the EATCS bulletin where they give an algorithm which allows to solve a number of elementary problems on context-free grammars including identifying useless variables, and deciding emptiness or finiteness.
The algorithm is a “saturation” algorithm that takes as input a finite state automaton \(\auto=(Q,\delta,q_I,F)\) and a context-free grammar \(\gram=(V,P)\) on an alphabet \(\Sigma\).
It adds transitions to \(\auto\) so that, upon saturation, its resulting language corresponds to the following set of words on alphabet \((V\cup \Sigma)\): \[ \{ \alpha \in (V\cup \Sigma)^{*} \mid \exists w \in L(\auto) \colon \alpha \Rightarrow^* w \}\enspace.
\]
This set, which they denote as \(\text{pre}^*(L(\auto))\), comprises the sentences of the grammar \(\gram\) which, upon application of zero or more production rules, becomes a word in \(L(\auto)\).
At the logical level, the algorithm consists of one saturation rule and nothing more which states that if \(X\rightarrow \beta\in P\) and \(q \xrightarrow{\beta}\!
\!^* q'\) then add a transition \( q \stackrel{X}{\rightarrow} q'\) to \(\auto\).

We show that this approach is relevant to our problem.
Let us consider the inclusion problem CFG into REG where we have the restriction that for every \(i\in[1,n]\) we have that \( L(\gram)_i \) is either a singleton word or it is empty: \(L(\gram)_i\neq\emptyset \implies L(\gram)_i\in\Sigma^*\).
Such grammars are called straight line programs (SLP) in the literature.
As we have shown in the past \cite{gantyRegularExpressionSearch2019}, the inclusion problem for straight line programs has direct applications in regular expression matching for text compressed using grammar-based techniques.
Intuitively, the regular expression is translated into \(L(\auto)\) and the compressed text is given by \(L(\gram)\).
Then we have that the regular expression has a match in the decompressed text if{}f \(L(\gram) \subseteq L(\auto)\).

In the context of our inclusion algorithm, the assumption that \(\gram\) is a straight line program ensures that, for every \(m\), each entry in the vector \(F_\mathcal{\gram}^m(\vec{\emptyset})\) is either the empty set or a singleton word.
Now consider the state-based variant of Algorithm~1 discussed above where words are replaced by their contexts.
Because of the SLP assumption we have that each entry of the \(\mathrm{cur}\) and \(\mathrm{prev}\) vectors in Algorithm~\ref{algo} stores a set of pairs of states: each entry is given by a set \(P \subseteq Q^2\) as opposed to \(P \subseteq \wp(Q^2)\) when \(\gram\) need not be a SLP.

A possible data-structure to store such a set \(P\) is a graph where nodes are given by \(Q\) and whose edges coincides with the pairs of states of \(P\).
Add labels on edges and you can encode a vector of sets like \(P\), namely an edge from \(q\) to \(q'\) with label \(m\) encodes that the pair \( (q,q') \) belongs to the \(m\)th component of the vector.
Next, observe that the updates to be carried out on the vectors \(\mathrm{cur}\) and \(\mathrm{prev}\) via \(F_{\gram}\) as prescribed by Algorithm~\ref{algo} can be implemented via the saturation rule described above.
For instance, assume \(F_{\gram}\) requires to update the \(m\)-th entry of \(\mathrm{cur}\) following the production rule \( X_m \rightarrow X_k\, X_{k'} \).
The corresponding update to the graph-based data structure adds an edge labelled \(m\) between states \( (q,q') \) if there exists a node \(q''\) such that \( (q,q'') \) and \( (q'',q') \) are two edges respectively labelled by \(k\) and \(k'\).
It is worth pointing that this update rule coincides with the above saturation rule for the rule \(X_m \rightarrow X_k\, X_{k'}\).
Therefore the saturation rule together with the underlying graph data-structure leveraging the automaton \(\auto\) enable a state-based implementation of Algorithm~\ref{algo} tailored to the case SLP into REG.
It is worth pointing out that Esparza et al.~\cite{esparzaUniformFrameworkProblems2000} work out the precise time and space complexity of implementing the saturation rule.
Ultimately, we obtain an algorithm in PTIME for the SLP into REG inclusion problem.
A more detailed complexity analysis can be found in Valero's PhD thesis \cite[Chapter~5]{ValeroMeja}.

\section{Other Language Inclusion Problems}
\label{sec:ouverture}

In this section we talk about the other language classes to which our framework can be applied.

\subsection{REG \(\subseteq\) REG}
We consider the case where \(\gram\) is a right regular grammar i.e. a grammar in which all production rules are of the form \(X_{j} \rightarrow X_{k}\, a\), \(X_{j} \rightarrow a\), or \(X_{j}\rightarrow \varepsilon\) where \(X_{j}\) and \(X_{k}\) are variables and \(a\in \Sigma\).
The function \(F_{\gram}\) of the least fixpoint characterization of \(L(\gram)\) is now defined as follows.
Let \(\Lang=(L_1,\ldots,L_n)\in \wp(\Sigma^*)^{n}\) define:
\[F_{\gram}(\Lang)_j \ud \bigcup_{X_{j}\rightarrow X_{k}\,a\in P}
	L_{k}\{a\} \cup \bigcup_{a \in \Sigma \cup \{\varepsilon\},~X_{j}\rightarrow a\in P}\set{a}\enspace .
\]
The same algorithm with this new \(F_{\gram}\) is also correct.
Notice however that for Proposition~\ref{prop:antichain_opt} and Proposition~\ref{prop:cv_check} to be correct it is now sufficient for \(\ltimes\) to only be right-monotonic.
This is because the new \(F_{\gram}\) only uses right-concatenations in its production rules.
As a consequence of this, we can relax \(\leq_{\text{ctx}}^{\auto}\) into the coarser quasiorder \(\leq_{\text{post}}^{\auto}\) where, we just replace the context \(\text{ctx}\) of a word \(u\) by the set \(\text{post}\) containing all the states reachable with the word \(u\) from the initial state of \(\auto\), that is, \(\set{q\mid q_I\xrightarrow{u}\!
	\!^* q'}\).
Similarly, we can relax \(\leq_{\text{ctx}}^{M}\) into the Nerode quasiorder of \(M\).
In the Nerode quasiorder, a word \(u\) subsumes \(v\) if, for every \(w\in \Sigma^{*}\), \(uw\in M\) implies \(vw\in M\).

\subsection{REG \(\subseteq\) Petri Net Traces}
In 1999 Petr Jančar, Javier Esparza and Faron Moller \cite{jancarPetriNetsRegular1999} published a paper which solved several language inclusion problems including the problem asking whether the trace set of a finite-state labelled transition system is included in the trace set of a Petri net (which defines a labelled transition system with possibly infinitely many states).
In this section we give a decision procedure following the framework we put in place in the previous sections.
As we will see our decision procedure is very close yet slightly different from the one they proposed.

It is worth pointing that their exposition favors a process theory viewpoint rather than an automata theory viewpoint which is why they talk about trace sets and labelled transition systems (or processes) instead of languages and acceptors.
With this view in mind, the notion of alphabet \(\Sigma\) becomes a set of \emph{actions} in which they include a distinguished \emph{silent} action \(\tau\).
The \(\tau\) resembles the \(\varepsilon\) symbol used in the automata theoretic setting.
In their paper, they define and prove correct an algorithm to decide whether the trace set given by a finite-state labelled transition system is included into the trace set given by a Petri net.

Our first step is to extract from their paper an ordering on words that is \(M\)-suitable.
In our setting, \(M\) is the trace set given by a Petri net.
The ordering is a “state-based” ordering defined upon the underlying Petri net.
The definition relies on the notion of \(\omega\)-markings which, intuitively, are vectors with as many entries as there are places in the Petri net and such that each entry takes its value in \(\mathbb{N} \cup \{\omega\}\).
Let \(P\) be the set of places of the Petri net, then \(\mathbb{N}^P_{\omega}\) denote the set of \(\omega\)-markings thereof.
Given two \(\omega\)-markings \(\hat{M}\) and \(\hat{M}'\) they define (p.~481) an ordering \(\mathord{\leqslant}\) between them such that \(\hat{M} \leqslant \hat{M}'\) if \(\hat{M}(p) \leqslant \hat{M}'(p)\) for all \(p\in P\) such that \(\omega\) is greater than every natural.
They then lift the ordering \(\mathord{\leqslant}\) to sets of \(\omega\)-markings as follows.
Given two sets \(\mathcal{M},\mathcal{M'} \subseteq \mathbb{N}^P_{\omega}\) define \(\mathcal{M} \leqslant \mathcal{M'}\) if{}f for each \(M\in\mathcal{M}\) there exists \(M'\in\mathcal{M'}\) such that \(M \leqslant M'\).
This relation is a quasiorder on sets of \(\omega\)-marking which becomes a well-quasiorder when restricted to the finite sets of \(\omega\)-markings.
That is \(\mathord{\leqslant}\) is a wqo on \(\wp_{\text{fin}}(\mathbb{N}^P_{\omega})\) (their Corollary~2.17).

We are now in position to define our ordering on words (what they call traces).
Let \(u,v\in (\Sigma\setminus\{\tau\})^*\) define \[ u \lessdot v \text{ if{}f } \max(\mathcal{C}(\mathord{\stackrel{u}{\Rightarrow}}(\{M_0\}))) \leqslant\max(\mathcal{C}(\mathord{\stackrel{v}{\Rightarrow}}(\{M_0\})))\enspace.
\]
Let us discuss this definition starting with the expressions \(\max(\mathcal{C}(\mathord{\stackrel{u}{\Rightarrow}}(\{M_0\})))\) which denote finite sets of \(\omega\)-markings (i.e. \(\max(\mathcal{C}(\mathord{\stackrel{u}{\Rightarrow}}(\{M_0\})))\in\wp_{\text{fin}}(\mathbb{N}_{\omega}^P)\) so that the ordering \(\mathord{\leqslant}\) is the well-quasiorder described above).
The details of the definition \(\max(\mathcal{C}(\mathord{\stackrel{u}{\Rightarrow}}(\{M_0\})))\) would take too much space to include here but intuitively the finite set of \(\omega\)-markings is a finite description of the possibly infinitely many markings (which are \(\omega\)-markings with no entry set to \(\omega\)) the Petri net can reach guided by the sequence of actions in \(u\) from its initial marking \(M_0\) (the \(\mathcal{C}\) and the \(\max\) operator turn a possibly infinite sets of markings into a finite set of \(\omega\)-markings).

Let us now turn to the \(M\)-suitability of \(\mathord{\lessdot}\).
As we have already shown the ordering \(\mathord{\lessdot}\) is a wqo.
Second it is routine to check that monotonicity follows from their Lemma~2.5(a) and Lemma~2.9 \cite{jancarPetriNetsRegular1999}.
Indeed Lemma~2.5(a) states that if \(\mathcal{M}\leqslant\mathcal{M'}\) then \(\mathord{\stackrel{a}{\Rightarrow}}(\mathcal{M}) \leqslant \mathord{\stackrel{a}{\Rightarrow}}(\mathcal{M'})\) for every \(a\in\Sigma\setminus\{\tau\}\); while Lemma~2.9 states that \(\mathcal{M}\leqslant\mathcal{M'}\) then \(\max(\mathcal{C}(\mathcal{M})) \leqslant \max(\mathcal{C}(\mathcal{M'}))\).

Next we show that \(M\)-preservation holds by simply observing that \(w\) is a trace of the Petri net if{}f \(\mathord{\stackrel{w}{\Rightarrow}}(\{M_0\})\neq\emptyset\), hence that if \(u \lessdot v\) and \(u\) is a trace then \(v\) is a trace by definition of \(\mathord{\lessdot}\), their Lemma~2.7(a) showing \(\mathcal{M}\subseteq\mathcal{C}(\mathcal{M})\), and the definition of \(\max\) (p.~482).

It remains to show that \(\mathord{\lessdot}\) is decidable which is easily established by first using the result of their Lemma 2.15 stating that we can effectively construct \(\max(\mathcal{C}(\mathord{\stackrel{a}{\Rightarrow}}(\{\mathcal{M}\}))\) given a finite set \(\mathcal{M}\) of \(\omega\)-markings and an action \(a\in\Sigma\); and second by using the result of their Lemma~2.14 to lift (inductively) their effectivity result from actions to finite sequences of actions.
Notice that this effectivity result for finite sequences of actions also gives a decision procedure for the membership in the trace set of the Petri net.

With all the ingredients in place, we obtain an algorithm deciding whether the trace set of a finite-state labelled transition system is included into the trace set of a Petri net.
It is worth pointing out that their algorithm \cite[Theorem~3.2]{jancarPetriNetsRegular1999} performs less comparisons of finite sets of \(\omega\)-markings than ours.
A tree structure restricts the comparisons they do during the exploration to pairs of sets of \(\omega\)-markings, provided one is the ancestor (w.r.t.\ tree structure) of the other.
Since our algorithm does not have such a restriction, we claim that we are making at least as many comparisons as their algorithm.
This potentially results in a shorter execution time and less membership checks in our algorithm.

\subsection{Languages of Infinite Words}
In this section we give the idea on how we extended the framework to the inclusion problem between languages of infinite words~\cite{DGPR21}.
In particular we studied the inclusion between two \(\omega\)-regular languages and the inclusion of an \(\omega\)-context-free language into an \(\omega\)-regular language.

Consider the inclusion problem \(L^{\omega}(\auto)\subseteq M\) where \(\auto\) is a BA and \(M\) is an \(\omega\)-regular language.
The first step is to reduce the inclusion check to the ultimately periodic words of the languages, as they suffice to decide the inclusion \cite{calbrixUltimatelyPeriodicWords1994}.
These are infinite words of the form \(uv^{\omega}\) for \(u\in \Sigma^{*}\) a finite prefix and \(v\in \Sigma^{+}\) a finite period.
We compare two ultimately periodic words using a pair of quasiorders \(\ltimes_{1}\) and \( \ltimes_{2}\) on finite words.
The first quasiorder, \(\ltimes_{1}\), compares the prefixes of ultimately periodic words, while the second quasiorder, \( \ltimes_{2}\), compares their periods.
As expected the quasiorders need to be \(M\)-preserving right-monotonic decidable wqos.
In this setting \(M\)-preservation means that if \(uv^{\omega}\in M\) and \((u,v) \ltimes_{1} \times \ltimes_{2} (u',v')\) then \(u'v'^{\omega}\in M\).
For example, by taking \(\ltimes_{1}\) to be \(\leq_{post}^{\mathcal{B}}\) and \(\ltimes_{2}\) to be \(\leq_{\text{ctx}}^{\mathcal{B}}\), where \(\mathcal{B}\) is a BA such that \(L^{\omega}(\mathcal{B})=M\), we have an \(M\)-preserving pair of right-monotonic decidable wqos.

An \(M\)-preserving pair of wqos ensures the existence of a finite subset \(S=S_{1}\times S_{2}\) of ultimately periodic words that is sufficient to decide the inclusion.
In order to compute such a subset \(S\) we establish a fixpoint characterization for the sets of prefixes and periods of \(L\).
Finally, we decide the inclusion by checking membership in \(M\) for every ultimately periodic word \(uv^{\omega}\) such that \((u,v)\in S\).

In the case where \(\auto\) is a Büchi Pushdown (and therefore where \(L^{\omega}(\auto)\) is an \(\omega\)-context-free language), the only difference lies in the fixpoint functions which now use both right and left concatenations.
Thus the pair of quasiorders in this case should be monotonic.
For example we can take both orders to be \(\leq_{\text{ctx}}^{\mathcal{B}}\).

\subsection{Pointers to Other Cases}
In our paper~\cite{doveriAntichainsAlgorithmsInclusion2023}, the approach is adapted to the inclusion problem between Visibly Pushdown Languages (VPL) (of finite words) and \(\omega\)-VPL (of infinite words), classes of languages defined in~\cite{AlurM04}.
In these cases, defining monotonicity conditions is challenging.

In the work by Henzinger et al.~\cite{nico}, the authors adapt our framework to solve their inclusion problem between operator-precedence languages~\cite{Floyd63}.
These languages fall within a class that is strictly contained in deterministic context-free languages and, in turn, strictly contains VPL~\cite{CRESPIREGHIZZI20121837}.

In~\cite{doveriFORQBasedLanguageInclusion2022}, we further extend the framework for solving the inclusion \(L^{\omega}(\auto)\subseteq M\), for \(\auto\) an BA and \(M\) an \(\omega\)-regular language, using a family of quasiorders instead of a pair of quasiorders.
A family of quasiorders allows more pruning when searching for a counterexample, thus lesser membership queries at the end.

As we show in~\cite{GantyRV2021} it is also possible to define an \(M\)-suitable quasiorder leveraging pre-computed simulation relations on the states of an automaton for \(M\).
This quasiorder might be coarser than the state-based one presented in Section~\ref{sec:example}.

Some directions for future work are the inclusion of a context-free language into a superdeterministic context-free language~\cite{greibachSuperdeterministicPDAsSubcase1980}, and the inclusion between tree automata.
It is worth noting that least fixpoint characterizations exist for a very large number of language classes \cite{Sorin82}.
Eventually, we want to explore whether our approach can be adapted to the emptiness problem of alternating automata for finite and infinite words.
The PhD thesis of Nicolas Maquet \cite{maquetNewAlgorithmsData2011} suggests it can be done.
In \cite{bonchiCheckingNFAEquivalence2013}, Bonchi and Pous compare antichains and bisimulation up-to techniques.
We want to extend this comparison,  given the novel insights of our approach.

	{
		\fontsize{9pt}{10.8}
		\selectfont
		\paragraph{Acknowledgments.}
		Pierre visited Javier during his first year of PhD, a visit that turned out to be a milestone in Pierre's career and has had influence up to this day.
		This visit also got Pierre a new colleague, a mentor and, most importantly, a friend.
		Pierre wishes to thank Javier from the bottom of his heart for all the good memories throughout the years.
		Chana is grateful to have had Javier as a PhD advisor and as an academic role model. She learned a lot from him, and benefited from the kind and studious atmosphere that he has established in his Chair at the Technical University of Munich.
		We also are thankful to the reviewers for their valuable feedback.
		This publication is part of the grant PID2022-138072OB-I00, funded by MCIN, FEDER, UE and has been partially supported by PRODIGY Project (TED2021-132464B-I00) funded by MCIN and the European Union NextGeneration.

		\paragraph{Disclosure of Interests.}
		The authors have no competing interests to declare that are relevant to the content of this article.
	}

\bibliographystyle{splncs04}
\bibliography{main}

\end{document}